\documentclass[11pt]{article}
\usepackage[letterpaper,margin=1in]{geometry}
\usepackage{times}
\usepackage{hyperref}
\usepackage{amsmath,amssymb,amsthm} 
\usepackage{cite}
\hypersetup{
pdfborder = 0 0 0,
colorlinks=true,
citecolor=blue,
}
\usepackage{graphicx}

\newtheorem{theorem}{Theorem}
\newtheorem{lemma}[theorem]{Lemma}

\newtheorem{definition}[theorem]{Definition}

\newtheorem{claim}[theorem]{Claim}

\newcommand{\nmin}{\ensuremath{\widetilde{N}}}

\usepackage{multicol}
\usepackage{wrapfig}
\usepackage{algorithm}
\usepackage{algpseudocode}

\begin{document}

\begin{titlepage}

  \title{The Computational Power of Beeps}
  \author{Seth Gilbert\thanks{Supported in part by NUS FRC T1-251RES1404}\\ National University of Singapore\\ {\tt seth.gilbert@comp.nus.edu.sg}  \and Calvin Newport\thanks{Supported in part by NSF grant CCF 1320279}\\ Georgetown University\\ {\tt cnewport@cs.georgetown.edu}}
  \date{}

  \maketitle

\begin{abstract}
In this paper, we study the quantity of computational resources (state machine states and/or probabilistic transition precision)
needed to solve specific problems 
in a single hop network where nodes communicate using only beeps.
We begin by focusing on randomized leader election.
We prove a lower bound on the states required to solve this problem with a given error bound, probability precision, and  (when relevant) network size lower bound.
We then show the bound tight with a matching upper bound.
Noting that our optimal upper bound is slow,
we describe two faster algorithms that trade some state optimality to gain efficiency.
We then turn our attention to more general classes of problems
by proving that
 once you have enough states to solve leader election with a given error bound,
you have (within constant factors) enough states to simulate correctly, with this same error bound, a logspace TM with a constant
number of unary input tapes: allowing you to solve a large and expressive set of problems.
These results identify a key simplicity threshold beyond which useful distributed computation is possible
in the beeping model.
\end{abstract}

\setcounter{page}{0}
\thispagestyle{empty}

\end{titlepage}


\section{Introduction}
\label{sec:intro}


The beeping model of network communication~\cite{CornejoK10,Afek14012011,AfekABCHK11,ScottJ013,AfekABCHK13,ForsterSW14} 
assumes a collection of computational {\em nodes},  connected in a network, that 
interact by {\em beeping} in synchronous rounds. 
If a node decides to beep in a given round, it receives no feedback from the channel.
On the other hand, if a node decides to listen, it is able to differentiate
between the following two cases: (1) no neighbor in the network topology beeped in this round, and (2) one or more neighbors beeped.

Existing work on this model provide two motivations.
The first concerns digital communication networks  (e.g.,~\cite{degesys:2007,CornejoK10}). 
Standard network communication (in which nodes interact using error-corrected packets
containing many bits of information) requires substantial time, energy, and computational overhead (at multiple stack layers)
to handle the necessary packet encoding, modulation, demodulation, and decoding. 
Beeps, on the other hand, provide an abstraction capturing the simplest possible communication primitive: a detectable burst
of energy.
In theory, beep layers could be implemented using a fraction of the complexity required by standard packet communication, 
establishing the possibility of {\em micro-network} stacks for settings where high speed and low cost are crucial.
The second motivation for the beeping model concerns a connection to biological systems (e.g.,~\cite{navlakha:2014,Afek14012011,ScottJ013}). 
Network communication in nature is often quite simple; e.g., noticing a flash of light 
from nearby fireflies or detecting a chemical marker diffused by nearby cells. Therefore, understanding how to achieve distributed
coordination using such basic primitives can provide insight into how such coordination arises in nature (see~\cite{navlakha:2014} for a recent survey of this approach).

\paragraph{A Key Question.}  
As detailed below, existing work on the beeping model seeks to solve useful problems as {\em efficiently} as possible in this primitive network setting.
In this paper, by contrast, we focus on solving useful problems as {\em simply} as possible (e.g., as measured by factors such as the size of the algorithm's state machine representation),
asking the key question: is it possible to solve problems with both simple communication {\em and} simple algorithms?
Notice, the answer is not {\em a priori} obvious. It might be the case, for example, that complexity is conserved,
so that simplifying the communication model requires more complex algorithms.
Or it might be the case that simple algorithms coordinating with beeps are sufficient for even complex tasks.
Given the above motivations for studying beeps, answering this question is crucial, as it will help us probe the feasibility of
useful networked systems---be them constructed by engineers or evolution---that are truly simple in both their communication methods and control logic.

%
%

\paragraph{Our Answers.}
We study a collection of $n$ nodes connected in a {\em single hop} topology (i.e., the network graph is a clique).
We model  the randomized algorithmic process executing on each node as a probabilistic state machine.
The two parameters describing the complexity of these algorithms are: (1) an upper bound on the number of states (indicated by integer $s \geq 1$);
and (2) an upper bound on the precision of the probabilistic transitions (indicated by integer $q \geq 2$, where we allow probabilistic
transitions to be labeled with probability $0$, $1$, or any value in the interval $[\frac{1}{q}, 1- \frac{1}{q}]$).
We ask how large these values must grow to solve specific problems. Our motivating premise is that smaller values imply simpler algorithms.
(Notice, by considering both $s$ and $q$, we can capture the trade-off between memory and probabilistic precision; a question of 
standalone interest; c.f.,~\cite{lenzen:2014}).

We begin by considering {\em leader election},
a fundamental primitive in distributed systems.
We prove that for a given error bound $\epsilon \in [0,1/2]$ and probabilistic precision $q$,
any algorithm that guarantees to solve leader election with probability $1-\epsilon$
requires $s=\Omega(\log_q{(1/\epsilon)})$ states.
Provided a lower bound $\nmin$ on the size of the network,
this lower bound {\em reduces} to $s=\Omega(\log_q{(1/\epsilon)}/\nmin)$ states.
That is, the more nodes in the network, the fewer states each node needs to solve the problem.

This lower bound leverages a reduction argument.
We begin by defining and lower bounding a helper problem called {\em $(1,k)$-loneliness detection},
which requires an algorithm to differentiate between $n=1$ and $n\geq k$
(but has no requirements for intermediate network sizes).
This bound uses an indistinguishability argument regarding how nodes move through a specified state sequence.
We then show how to transform a solution to leader election for size lower bound $\nmin$,
to solve $(1,\nmin)$-loneliness detection---allowing our loneliness bound to carry over to leader election.

We then turn our attention to leader election upper bounds.
We begin by proving our lower bound tight by showing, for every network size lower bound $\nmin \geq 1$,
how to solve leader election with $s=O(\log_q{(1/\epsilon)}/\nmin)$ states.
The key idea behind this algorithm is to have nodes work together to implement a distributed timer.
The more nodes in the network, the longer the distributed timer runs, 
and the longer the distributed timer runs, the higher the probability that we succeed at leader election.
In this way, increasing the network size reduces the states required to hit a specific error bound.
A shortcoming of this new algorithm, however, is that its expected running time is exponential in the network size.
With this mind, we then describe two faster algorithms (their time is polylogarithmic in the relevant parameters)
that require only the minimum precision of $q=2$.
The cost for their efficiency, however, is a loss of state optimality in some circumstances.

The first algorithm requires $s=O(\log{(1/\epsilon)})$ states and solves leader
election with probability at least $1-\epsilon$, for any network size $n$.
It terminates in $O(\log{(n + 1/\epsilon)}\log{(1/\epsilon}))$ rounds, with probability at least $1-\epsilon$.
The key idea behind this algorithm is to test a potentially successful election by having the potential
leader(s) broadcast with probability $1/2$ for $\log{(1/\epsilon})$ rounds, looking for evidence
of company. It is straightforward to see that a single such test fails
with probability no more than $(1/2)^{\log{(1/\epsilon)}} = \epsilon$.
The problem, however, is that as the network size grows, the number of such tests performed also increases,
making it more likely that one fails. We neutralize this problem in our analysis by showing that 
the test failure probabilities fall away as a geometric series in the test count---bounding the cumulative error sum as the network grows.

The second algorithm requires only $s=O(1)$ states,
and yet, for every network size $n$,
it solves leader election with high probability in $n$ when run in a network of that size.
It requires only $O(\log^2{n})$ rounds, with high probability.  
%
The key idea driving this algorithm is to harness the large amount
of total states in the network to implement a distributed timer that requires
$\Theta(\log{n})$ time to countdown to $0$, when executed among $n$ nodes.
This duration is sufficient for the nodes to safely reduce contention down to a single leader.

After studying leader election,
we turn our attention to more general classes of distributed decision problems.
Leveraging our leader election algorithms as a key primitive,
we show how to simulate a logspace decider Turing Machine (TM)
with a constant number of unary inputs (all defined with respect to the network size $n$). 
Perhaps surprisingly, this algorithm requires only $O(\log{(1/\epsilon)})$ states to complete
the simulation with probability $1-\epsilon$, and only $O(1)$ states to achieve high probability in $n$.
(Notice that this is not enough states for an individual node to store even a single
pointer to the tape of the simulated machine.)
Our simulation
uses the same general strategy first highlighted in the study of population protocols~\cite{AngluinADFP06}: simulate a counter machine 
with a constant number of counters that hold values from $0$ to $O(n)$,
and then apply a transformation due to Minsky~\cite{minsky1967} to simulate
a logspace TM with this machine.
Due to the differences between the  beeping and population protocol models,
however, our counter machine simulation strategies are distinct from~\cite{AngluinADFP06}.

\paragraph{Implications.}
The results summarized above establish that the $\log{(1/\epsilon)}$ state threshold for leader election with bounded error
is (in some sense) a fundamental simplicity threshold for solving useful problems with beeps.
It is striking that
if you have {\em slightly less} than this much memory, even the basic symmetry breaking task of leader election is impossible,
but
if you instead have {\em slightly more}, then suddenly you can solve large classes of complicated problems (i.e., everything solvable by a logspace TM).
If you are satisfied with high probability solutions (which is often the case),
then this treshhold reduces even more all the way down to $O(1)$.
Given these results, we tentatively claim a positive answer to the key question posed above:  {\em complexity is not destiny;
you can solve hard problems simply in simple network models.}

\bigskip

\noindent Before proceeding into the technical details of our paper, we will first take the time to place both our model and our
results in the context of the several different areas of relevant related work. Among other questions, we want to understand
the relationship of our bounds to existing beep results, and how the beeping model compares and contrasts
to similar settings.

\paragraph{Comparison to Existing Beep Results.}
The algorithmic study of beeping networks
began with Degesys et~al.~\cite{degesys:2007}, who
introduced a continuous variant of the beeping model, inspired by the pulse-coupled oscillator framework.
They studied biologically inspired strategies for solving a {\em desynchronization} problem.  
Follow-up work generalized the results to multihop networks~\cite{degesys:2008,MotskinRSG09}.
Cornejo and Kuhn~\cite{CornejoK10} introduced the discrete (i.e., round-based) beeping model studied in this paper.
They motivated this model by noting the continuous model in~\cite{degesys:2007,degesys:2008,MotskinRSG09} was unrealistic and yielded trivial solutions to desynchronization,
they then demonstrated how to solve desynchronization without these assumptions.
Around this same time, Afek et~al.~\cite{Afek14012011} described a maximal independent set (MIS) algorithm in a strong
version of the discrete beeping model.
They argued that something like this algorithm might play a role in the
proper distribution of
sensory organ precursor cells in fruit fly nervous system development.
Follow-up work~\cite{AfekABCHK11,ScottJ013,AfekABCHK13} removed some of the stronger assumptions of~\cite{Afek14012011} and improved the time complexity.
In recent work,  F{\"{o}}rster et~al.~\cite{ForsterSW14} considered deterministic leader election in a multihop beeping network.

To place this paper in this context of the existing work on the beeping model, it is important to note that the above-cited
papers focus primarily on two goals: minimizing time complexity and minimizing information provided to nodes  (e.g., network size, max degree, global round counter).
They do not, however, place restrictions on the amount of states used by their algorithms. 
Accordingly, these existing results require either:
the ability to store values as large as $\Theta(n)$~\cite{CornejoK10,Afek14012011,AfekABCHK11,ScottJ013,AfekABCHK13},
 or uniques ids~\cite{ForsterSW14} (which in our framework would require a machine with $n$ different initial states,
or equivalently, $n$ different machines).
In this paper, 
we prove that the algorithmic complexity threshold for solving many useful problems is actually much lower:
$O(1)$ states are sufficient for high probability results and $O(\log{(1/\epsilon)})$ states are sufficient for fixed error bound results.\footnote{Notice, direct comparisons between many of these results is complicated by the variety of possible assumptions; e.g., synchronous versus asynchronous starts,
multihop versus single hop, small versus large probability precision.}
We argue the direction pursued in this paper (how complex must algorithms become to solve useful problems with beeps)
complements the direction pursued in existing papers (how fast can algorithms solve useful problems with beeps).
Answers to both types of queries is necessary to continue to understand the important topic of coordination in constrained network environments.

\paragraph{Comparison to the Radio Network Model.}
The standard radio network model allows nodes to send large messages,
but assumes concurrent transmissions lead to message loss (that may or may not be detectable).
The key difference between the radio network model and the beeping model
is that in the former you can recognize the case where exactly one node broadcast (e.g., because you receive a message). 
This capability, which the beeping model does not offer (a single beeper looks like multiple beepers),
is powerful. It allows, for example, algorithms that can solve leader election with deterministic safety
using only a constant amount of state,
when run in network of size at least $2$.
If you assume receiver collision detection, these solutions require only polylogarithmic expected time.\footnote{For example: 
divide rounds into pairs of even and odd rounds.
In even rounds, nodes broadcast a simple message with constant probability. If a node ever succeeds in broadcasting
alone, all other nodes become {\em heralds}. They stop competing in even rounds and begin competing in odd rounds.
When the winner (who is now the only non-herald in the network) eventually hears a message in an odd round,
it elects itself leader. If we assume collision detection, we can reduce contention fast in the even  rounds
with basic knockout protocols; e.g., if you choose to listen and detect a collision you are knocked out and just wait to become a herald.}
These results violate our lower bounds for leader election with beeps (where the state size grows toward infinity as you drive the error bound toward $0$)---indicating that the communication limitations
in the beeping model matter from a computability perspective.

\paragraph{Comparison to the Stone Age Computing Model.}
It is also important to place our results in the context of other simplified communication/computation models.
Consider, for example, the stone age distributed computing model introduced by Emek and Wattenhofer~\cite{Emek:2013}.
This model assumes state machines of constant size connected in a network and executing asynchronously.
The machines communicate with a constant-size message alphabet and when transitioning 
can distinguish between having received $0$, $1$, or $\geq b$ messages of each type,
for some constant parameter $b \geq 1$.
For $b=1$, this model is essentially an asynchronous version of the beeping model.
To this end, nodes in our model can simulate nodes in the stone age model with $b=1$ indefinitely using a constant amount of states.
For $b>1$, however, any such simulation likely becomes impossible in the beeping model with a constant amount of states.
As noted in our discussion of the radio network model, the ability to safely recognize the case of exactly one message being sent provides
extra power beyond what is achievable (without error) using only beeps. 

\paragraph{Comparison to the Population Protocol Model.}
Another relevant simplified communication/computation setting is the well-studied population protocol 
model~\cite{AngluinADFP06,AngluinAE08a,AngluinAE08,AngluinAER07,AngluinAE06,ChatzigiannakisS08} .
This model describes nodes as state machines of constant size that interact in a pairwise manner---transforming both states asymmetrically.
In the basic version of the model, a fair scheduler chooses pairs to interact. A version in which the scheduler is randomized
adds more power. There are similarities in the goals pursued by the beeping and population protocol models:
both seek (among other things) to understand the limits of limited state in distributed computation.
The core difference between the two settings is the role of the algorithm in communication scheduling.
In the beeping model, algorithms must reduce contention and schedule communication on their own.
In the population protocol model the scheduler ensures fair and reliable interactions.
Imagine, for example, a continuous leader election problem where every node has a {\em leader} bit,
and the problem
requires in an infinite execution that: (1) every node sets {\em leader} to $1$ an infinite number of times;
and (2) there is never a time at which two nodes both have {\em leader} set to $1$.
This problem is trivial in the population protocol: simply pass a leader token around the network.
In the beeping model, by contrast, it is impossible as it essentially 
requires nodes to solve leader election
correctly an infinite number of times---a feat which would require an unachievable error bound of $0$.
It follows that in some respects these two models are studying the impact of limited state on different
aspects of distributed computation.

\section{Model}

We model a collection of $n$ probabilistic computational agents (i.e., ``nodes") that
are connected in a single hop network and
communicate using a unary primitive; i.e., {\em beeps}.
They execute in synchronous rounds. In each round, each node can either beep or receive.
Receiving nodes can distinguish between the following two cases: 
(1) no node beeped; (2) one or more nodes beeped.
We characterize these agents by $s$ (a bound on the number of states in their state machine definition),
and $q$ (a bound on the precision allowed in probabilistic transitions, with larger values enabling more accurate transition probabilities).
We now  formalize these model definitions and assumptions.

\paragraph{Node Definition.}
We formalize the algorithm executing on each node as a probabilistic
state machine $M = (Q_r, Q_b,q_s,\delta_{\bot}, \delta_{\top})$,
where: $Q_r$ and $Q_b$ are two disjoint sets of states corresponding to receiving and beeping, respectively;
$q_s$ is the start state; and $\delta_{\bot}$ and $\delta_{\top}$ 
are the probabilistic transition functions\footnote{These transition functions map the current state to a distribution
over the states to enter next.} applied in the
case where the node detects silence and where the node beeps/detects a beep, 
respectively. 

Some problems have all nodes execute the same state machine,
while others include multiple machine types in the system, each corresponding to a different initial value.

\paragraph{Executions.}
Executions proceed in synchronous rounds with all nodes in the network starting in their machine's start state.
At the beginning of each round $r$, 
for each node $u$ running a machine $(Q_r, Q_b,q_s,\delta_{\bot}, \delta_{\top})$,
if its current state $q_u$ is in $Q_b$, then $u$ emits a beep, otherwise it receives.
If at least one node beeps in $r$, then it follows that {\em all} nodes either beep or detect a beep
in this round. Therefore, each node $u$ applies the transition function $\delta_{\top}$
to its current state $q_u$ and selects its next state according to the resulting distribution, $\delta_{\top}(q_u)$.
If no node beeps in $r$, then each node $u$ applies the transition function $\delta_{\bot}$,
selecting its next state from the distribution,  $\delta_{\bot}(q_u)$.

\paragraph{Parameters.}
We parameterize the state machines in our model with two values.
The first, indicated by $s \geq 1$, is an upper bound on the number of states allowed (i.e., $|Q_r| + |Q_b| \leq s$).
The second, indicated by $q \geq 2$, bounds the precision of the probabilistic transitions allowed by the $\delta$ functions.
In more detail, for a given $q$, the probabilities assigned to states by distributions in the range of $\delta$
must either be $0$, $1$, or in the interval, $[\frac{1}{q}, 1-\frac{1}{q}]$.
For the minimum value of $q=2$, for example, probabilistic transitions can only occur
with probability $1/2$.
As $q$ increases, smaller probabilities, as well as probabilities closer to $1$, become possible.
Finally, we parameterize a given execution with $n$---the number of nodes executing in the network.


\section{Leader Election}

The first computational task we consider is leader election: eventually, one node designates itself leader.  An algorithm state machine
that solves leader election must include a final {\em leader state} $q_{\ell}$ that is terminal (once a node enters the state, it never leaves). If a node enters this state it indicates the node has elected itself leader.
For a given error bound $\epsilon \in [0,1/2]$,
we say an algorithm {\em solves} leader election with respect to $\epsilon$ if when executed in a network of any size,
it satisfies the following two properties: (1) {\em liveness}: with probability $1$, at least one node eventually enters the leader state; and (2) {\em safety}:  with probability at least $1-\epsilon$, there is never more than $1$ node in the leader state.
We also consider algorithms for leader election that are designed for networks of some minimal size $\nmin$.
 In this case, the algorithm must guarantee liveness in every execution, 
 but it needs to guarantee safety only if the network size $n$ is at least $\nmin$.
Our goal is to develop algorithms that use a minimum number of states to solve leader election for a given error bound $\epsilon$,
 probability precision $q$, and, when relevant, network size minimum $\nmin$.

\paragraph{Roadmap.}
In Section~\ref{sec:leader:lower}, we present a lower bound for leader election.  In Section~\ref{sec:universal}, we present a universal algorithm template, followed by three specific instantiations in Sections~\ref{sec:leader:stateoptimal},~\ref{sec:fast}, and~\ref{sec:faster}.

\subsection{Leader Election Lower Bound}
\label{sec:leader:lower}

Here we analyze the number of states required to solve leader election given
a fixed $\epsilon$, $q$, and network size lower bound $\nmin$.
Our main result establishes that the number of states, $s$,
must be in $\Omega(\lceil \frac{\log_{q}{(1/\epsilon)}}{\nmin}\rceil)$.

To prove this result, we begin by defining and bounding a helper problem called {\em $(1,k)$-loneliness detection},
which requires an algorithm to safely distinguish between $n=1$ and $n \geq k$.
The bound leverages a probabilistic indistinguishability argument concerning a short
execution of the state machine in both the $n=1$ and $n=k$ cases.
We then show that loneliness detection captures a core challenge of leader
election by demonstrating how to transform a leader election algorithm
that works for $n \geq \nmin$ into a solution to $(1,\nmin)$-loneliness detection.
The bound for the latter then carries over to leader election by reduction. 

\paragraph{$(1,k)$-Loneliness Detection.}
The $(1,k)$-loneliness detection problem is defined for some integer
$k>1$ and error bound $\epsilon$. It assumes all nodes run the same state machine with two
special terminal final states that we label $q_a$ (indicating ``I am alone")
and $q_c$ (indicating ``I am in a crowd").
The {\em liveness} property of this problem requires
that  with probability $1$,
every node  eventually enters a final state.
The {\em safety} property requires that with probability at least $1-\epsilon$, the following holds:
if $n=1$, then the single node in the system eventually enters $q_a$; and if $n \geq k$
then all nodes eventually enter $q_c$.
Crucial to this problem definition is that we do not place any restrictions
on the final states nodes enter for the case where $1 < n < k$.

The following bound formalizes the intuition
that it becomes easier to break symmetry, and therefore easier
to solve loneliness detection, as the threshold for detecting
a crowd grows. Put another way:
the presence of a big crowd is easier to detect than a small crowd.

\begin{lemma}
Fix some integer $k>1$. Let ${\cal L}$ be an algorithm that solves
$(1,k)$-loneliness detection with error bound $\epsilon$ and probability precision $q$
using $s$ states.
It follows that $s = \Omega( \frac{\log_{q}{(1/\epsilon)}}{k})$.
\label{lem:lonely}
\end{lemma}
\begin{proof}
Fix some integer $k>1$, error bound $\epsilon$,
and probability precision $q$.
Fix some algorithm ${\cal L}$ that solves $(1,k)$-loneliness detection
with error probability $\epsilon$, using precision $q$.
Let $(Q_r, Q_b, q_s, \delta_{\bot}, \delta_{\top})$ be ${\cal L}$'s state machine description.
Let $q_{a}$ and $q_{c}$ be the two terminal final states
 required by the problem definition.
We note that when $q \geq (1/\epsilon)$, the lower bound claim on the state size
reduces to a constant (or smaller), which is trivially true.
The remainder of this proof therefore addresses the more interesting case where $q < (1/\epsilon)$. 
Our goal is prove the required lower bound on the number of states, $s$,
needed for ${\cal L}$ to solve loneliness detection given these fixed values of $k$, $\epsilon$, and $q$.

Our first step toward this goal is 
to introduce the notion of a {\em solo reachable path}, defined with respect to ${\cal L}$.
In more detail, we say a sequence $P = q_1,q_2,...,q_x$ of $x$ states from $Q_r \cup Q_b$
is a {\em solo reachable path} for ${\cal L}$'s state machine if and only if 
$q_1 = q_s$, and
for each consecutive pair of states $q_i,q_{i+1}$ in $P$, the following hold:

\begin{enumerate}
\item if $q_i \in Q_r$, then the probability assigned to $q_{i+1}$ in $\delta_{\bot}(q_i)$ is greater than $0$;

\item if $q_i \in Q_b$, then the probability assigned to $q_{i+1}$ in $\delta_{\top}(q_i)$ is greater than $0$.
%

\end{enumerate}

\noindent Put another way, a solo reachable path is a sequence of states that a node running
this machine might feasibly follow (i.e., it occurs with a non-zero probability)
 in a network with $n=1$. 

Fix any such solo reachable path $P=q_1,q_2,...,q_x$.
We will now consider the probability that a network consisting of exactly $k$
nodes follows this path.
In more detail, we claim that for every $r\in \{1,2,...,x\}$,
the probability that all $k$ nodes are in state $q_r$ (from $P$)
after $r$ rounds is at least $((1/q)^k)^{r}$.
We can prove this claim by induction on the state index $r$.

{\em Basis ($r=1$).} By definition, all nodes start in state $q_s$, so this occurs
with probability $1 > (1/q)^k$.

{\em Step.}
Assume the claim holds for some $r < x$. To show
it holds for $r+1$,
we note that the probability that any single node transitions
from $q_r$ to $q_{r+1}$ is greater than $0$ (by our definition of {\em reachable}).
The smallest probability greater than $0$ in our system is $1/q$.
Therefore, we can refine our statement to say this occurs with probability at least $(1/q)$.
It follows that the probability that {\em all $k$} nodes make the same transition is  at least $(1/q)^k$. 
Multiply this 
probability
by the probability $((1/q)^k)^r$ that all nodes followed $P$ up to $q_r$
(provided by the inductive hypothesis),
and we get the desired final probability of $((1/q)^k)^{r+1}$.

We next argue that there exists a useful solo reachable path that is not too long:

\begin{quote}
{\em (*) There exists a solo reachable path $P = q_1,q_2,...,q_x$, 
defined with respect to ${\cal L}$,
such that $q_1=q_s$, $q_x = q_{a}$, and $x \leq s$.}
\end{quote}

The fact that there exists a solo reachable path that begins in the start state $q_s$
and ends in the final state $q_{a}$,
follows from the safety property of loneliness detection,
which says that when $n=1$,
with probability at least $1-\epsilon > 0$, the single node, starting in state $q_s$,
ends up in final state $q_a$.
The fact that $x \leq s$ follows from the observation that if there is {\em any} such solo reachable path leading
$q_{a}$,
then we can excise the loops to get a path that remains reachable, but that
 never visits the same state more than once.

Let $P$ be the solo reachable path identified 
above by claim (*).
If we apply our inductive argument to $P$,
and leverage the fact that the length of $P$
is no more than $s$ (i.e., $x \leq s$),
we get that the probability that all $k$ nodes in a network
of $k$ nodes follow path $P$ is at least $(1/q)^{k\cdot s}$.
Notice, if this occurs, we have violated safety.
It must therefore hold that $(1/q)^{k\cdot s} \leq \epsilon$.
We can set up this requirement as a simple constraint that will provide the minimum allowable value for $s$:

\[ (1/q)^{k\cdot s} \leq \epsilon \Rightarrow  q^{k\cdot s} \geq (1/\epsilon) \Rightarrow s\cdot k \log{q} \geq \log{(1/\epsilon)} \Rightarrow s \geq \frac{\log{(1/\epsilon})}{k\log{q}} = \frac{\log_{q}{(1/\epsilon})}{k} .  \]

\noindent Combining these pieces, we have shown that if ${\cal L}$ solves $(1,k)$-loneliness detection,
than it must be the case that $s \geq \frac{\log_{q}{(1/\epsilon})}{k}$, as required by the lemma statement.

\end{proof}

\paragraph{Reducing Loneliness Detection to Leader Election.}
We now leverage the above result on $(1,k)$-loneliness detection
to prove a lower bound for leader election under the guarantee
that the network size $n \geq \nmin$.
The proof proceeds by reduction: we show how to transform such a leader
election solution into a loneliness detection algorithm of similar state size.

%
\begin{theorem}
Fix some network size lower bound $\nmin \geq 1$. 
Let ${\cal A}$ be an algorithm that solves leader election with error bound $\epsilon$ and probability precision $q$
using $s$ states
in any network where $n \geq \nmin$.
It follows that $s\in \Omega(\frac{\log_{q}{(1/\epsilon)}}{\nmin})$.
\label{thm:lowerleader}
\end{theorem}
\begin{proof}
Fix some algorithm ${\cal A}$
that solves leader election with error bound $\epsilon$ and probability precision $q$
using $s$ states
in a network where $n \geq \nmin$, for some integer $\nmin\geq 1$.
Our first step is to use ${\cal A}$ to create an algorithm ${\cal L}_{\cal A}$
that solve $(1,\nmin)$-loneliness detection for the same $\epsilon$ and $q$,
using $O(s)$ states.
Our new algorithm ${\cal L}_{\cal A}$ works as follows:

\begin{quote}
The algorithm partitions rounds into pairs.
During the first round of each pair, it executes ${\cal A}$.
The second round is used to announce the election of a leader. That is, if a node becomes leader according to ${\cal A}$,
it beeps in the second round. If any node beeps in a second round,
the algorithm stops its execution of ${\cal A}$ and moves onto the next phase.
This next phase consists of single round.
During this round, any node not elected leader beeps.
If this round is silent, then the algorithm enters final state $q_a$,
otherwise it enters state $q_c$.
\end{quote}

To analyze ${\cal L}_{\cal A}$, we first note that liveness follows directly from the liveness
guarantee of ${\cal A}$: once the simulation of ${\cal A}$ elects the leader,
${\cal  L}_{\cal A}$ will lead all nodes to the second phase where they
will then enter a final state after an single additional round.
We now consider safety. There are three relevant cases, depending on the value of $n$.

\begin{itemize}

\item {\em Case $\#1$:} $n=1$. In the case, the liveness guarantee of ${\cal A}$
(which holds regardless of the network size)
implies that the single node in the system will eventually become leader.
Because there are no other nodes in the system, the second phase round
will be silent. It follows that the single node will enter state $q_a$, as required.
It follows that safety is satisfied with probability $1$.

\item {\em Case $\#2$:} $n \geq \nmin$. In this case, by assumption, ${\cal A}$ correctly
solves leader election with probability at least $1-\epsilon$.
Assume this occurs and $u$ is the single leader elected.
If $n =1$ (which is possible when $\nmin=1$), the argument for Case \#1 applies,
and $u$ will enter $q_a$, as required.
Assume instead that $n > 1$.
In this case,  the second phase round will {\em not} be silent.
It follows that all nodes will enter state $q_c$, as required.
It follows that safety is satisfied with probability at least $1-\epsilon$.

\item {\em Case $\#3$:} $1 < n < \nmin$. In this case, there are no safety requirements.
Therefore, safety is vacuously satisfied with probability $1$.

\end{itemize}

We have just shown that given a solution to leader election that satisfies safety for $n \geq \nmin$,
we can solve $(1,\nmin)$-loneliness detection while growing the state size by at most a constant factor.
We can now pull together the pieces.
By Lemma~\ref{lem:lonely}, 
any solution to  $(1,\nmin)$-loneliness detection for a given $\epsilon$ and $q$,
requires a state size $s' = \Omega(\frac{\log_{q}{(1/\epsilon)}}{\nmin})$.
It follows that the states used by ${\cal A}$ to solve leader election for $\nmin$
can be no more than a constant
factor larger than $s'$, proving the same asymptotic bound on $s$, which
matches the theorem statement.
\end{proof}


\begin{wrapfigure}{R}{0.52\textwidth}
\begin{minipage}{0.52\textwidth}
\begin{algorithm}[H]
  \caption{Universal Leader Election Algorithm}      
       \begin{algorithmic}[1]
  \scriptsize
 \State $active \gets 1$
 \State $ko \gets 1$
 \State ${\hat q} \gets \min\{q, (1/\epsilon)\}$
 \State $done \gets$ {\bf \em [Termination Subroutine]}$(active, ko)$
 \State $ko \gets 0$
 \State 
 \While{({\bf not} $done$)}
 \State $participate \gets$ {\bf rbit}$()$ \Comment{Returns $0$ with prob $1/{\hat q}$, else $1$}
 \State $chan \gets \top$
 \State \Comment{Knock Out Logic}
 \If{$active \wedge participate$}
    \State {\bf beep}()
 \Else
   \State $chan \gets$ {\bf recv}()
 \EndIf
 \If{$active \wedge \text{\bf not } participate \wedge (chan = \top)$}
       \State $active \gets 0$
       \State $ko \gets 1$
 \EndIf
 \State \Comment{Termination Detection Logic}
 \If{$chan = \bot$}
	      \State $done \gets$ {\bf \em [Termination Subroutine]}$(active, ko)$
	      \State $ko \gets 0$
 \EndIf
 
 \EndWhile
 \State \Comment{Become Leader if Still Active}
 \If{$active$}
    \State $leader\gets 1$
 \Else
    \State $leader \gets 0$
 \EndIf
 \State {\bf return}$(leader)$

%
   \end{algorithmic}

       \label{alg:singlehop} 
      \end{algorithm}
\end{minipage}

%
\label{fig:leader}
\end{wrapfigure}

\subsection{The Universal Leader Election Algorithm}
\label{sec:universal}

We now turn our attention to leader election upper bounds.
The three results that follow adopt a template/subroutine approach. In more detail, Figure~\ref{fig:leader} describes what we call the {\em universal leader election} algorithm. 
This algorithm, in turn, 
makes calls to a ``termination subroutine." 
Different versions of this subroutine can be plugged into the universal algorithm, yielding different guarantees. 
Notice, this universal algorithm is parameterized with probability precision $q$ and error bound $\epsilon$,
which it uses to define the useful parameter $\hat q = \min\{q, (1/\epsilon)\}$.
This algorithm (as well as one of our termination subroutines) uses $1/{\hat q}$,
not $1/q$, as its smallest transition probability (intuitively, there is little advantage in using a probability too much smaller than our error bound $\epsilon$).

The basic operation of the algorithm is simple.  Every node is initially active.  Until the termination subroutine determines that it is time to stop, nodes repeatedly execute the knockout loop (lines 7--25).  In each iteration of the loop, each active node beeps with probability $1-1/\hat{q}$ and listens otherwise.  If a node ever hears a beep, it is knocked out, setting $ko=true$ and $active=false$.  In any silent iteration where no node beeps, they execute the termination subroutine to decide whether to stop.  Once termination is reached, any node that remains active becomes the leader.

\paragraph{Termination Subroutines.}
The goal of the termination subroutine is to decide whether leader election has been solved: it returns \emph{true} if there is a leader and \emph{false} otherwise.

The termination subroutine is called simultaneously by all the nodes in the system, and it is passed two parameters: the value of $active$, which indicates whether or not the calling node is still contending to become leader, and $ko$, which indicates whether or not it has
been knocked out in the main loop since the last call to the subroutine. 
We fix $R = 4\log_{\hat{q}}(\max(n, 1/\epsilon))$: a parameter, which as we will later elaborate,
captures a bound on the calls to the subroutine needed before likely termination.
We consider the following properties of a termination detection routine, defined with respect to error parameter $\epsilon$ and $R$:

\begin{enumerate} 
\item  \emph{Agreement}: Every node always returns the same value. 
\item \emph{Safety}: Over the first $R$ invocations, the probability that it returns true in any invocation with more than $1$ active node is at most $\epsilon/2$. 
\item \emph{Eventual Termination}: If it is called infinitely often with only one active node, then eventually (with probability 1), it returns true.
\item \emph{Fast Termination}: If it is called with only one active node, and with at least one node where $ko = true$, then it returns true.
\end{enumerate}

\paragraph{Universal Leader Election Analysis.} 
We now observe that the universal leader election algorithm is correct when combined with a termination subroutine that satisfies the relevant properties from above.
To do so,
we first determine how many rounds it takes until there is only one active node, and hence one possible leader.  We say that an iteration of the knockout loop (lines 7--25) is \emph{silent} if no node beeps during it.  (Notice that the termination routine is only executed in silent iterations of the knockout loop.) We first bound how long it takes to reduce the number of active nodes:
\begin{lemma}
\label{lem:ko}
Given probability $\epsilon \leq 1/2$ and parameter $R = 4\log_{\hat{q}}(\max(n, 1/\epsilon))$: 
after $R$ silent iterations of the knockout loop (lines 7--25), there remains exactly one active node, with probability at least $1-\epsilon/2$.
\end{lemma}
\begin{proof}
Notice that there is always at least one active node, because a node becomes inactive only upon detecting a beep from another active node.  
Fix two nodes $p$ and $q$.  For $p$ to remain active, node $q$ has to remain silent whenever $p$ listens. 
 Thus, if in some prefix of the execution, node $p$ listens $r$ times, the probability that both $p$ and $q$ remain active is at most $1/\hat{q}^r$.

Taking a union bound over all pairs of nodes, we can upper bound by $n^2/\hat{q}^R$ the probability that there exist any pair of nodes $p$ and $q$ where one of the nodes listens $R$ times and both nodes remain active.
Equivalently, we have shown that the probability of $R$ silent iterations of the knockout loop, while at least $2$  nodes are active, is bounded by $n^2/\hat{q}^R$.  

Now recall that $R = 4\log_{\hat{q}}(\max(n, 1/\epsilon))$.  Consider the two cases:
\begin{itemize}
\item If $\epsilon \geq 1/n$, then $R = 4\log_{\hat{q}}{n}$, and therefore: $n^2/\hat{q}^R \leq n^2/n^4 \leq \epsilon/2$.
\item If $\epsilon < 1/n$, then $R = 4\log_{\hat{q}}(1/\epsilon))$, and therefore: $n^2/\hat{q}^R \leq n^2/(1/\epsilon)^{4} \leq \epsilon^{2} \leq \epsilon/2$.
\end{itemize}
In both cases, within $R$ silent iterations, there is exactly one leader with probability at least $1-\epsilon/2$.
\end{proof}

Let $T$ be a termination subroutine that satisfies Agreement and Eventual Termination.  In addition, assume that $T$ satisfies safety in networks of size at least $\nmin$. We can now show that the universal leader election algorithm is correct with termination subroutine $T$:
\begin{theorem} \label{thm:universal}
If termination subroutine $T$ uses $s$ states and precision $q$, then the universal algorithm solves leader election with error $\epsilon$,
 $s + O(1)$ states, and $q$ precision  (guaranteeing safety only in networks of size $n \geq \nmin$).
\end{theorem}
\begin{proof} 
We consider the two properties of leader election in turn. 

\smallskip

\emph{Liveness:} First, there is always at least one active node, because a node becomes inactive only upon receiving a beep from another active node. In every iteration of the knockout loop that starts with more than one active node, there is a non-zero probability that at least one active node is knocked out, and therefore with probability one there is eventually only one active node remaining.  From that point on, in every iteration, there is a constant probability that the remaining active node is silent and hence the termination subroutine is executed by all nodes.  By the Termination property, eventually it will return true and all nodes terminate.  Thus with probability one, eventually there is at least one leader.

\emph{Safety:} Assume the network is of size at least $\nmin$ (the only case for which safety must hold). 
The probability that the algorithm makes it through $R$ silent iterations of the knockout loop with more than one leader is at most $\epsilon/2$, via Lemma~\ref{lem:ko}.  Notice that the termination routine is only executed in silent iterations of the knockout loop, and over these $R$ (potentially bad) iterations, there is a probability of at most $\epsilon/2$ of improperly terminating and entering the leader state with multiple active nodes (as provided
by the termination detection safety property). A union bound
combines these errors for a total error probability less than $\epsilon$.

\end{proof}

While the preceding theorem can be used to show the feasibility of solving leader election, it does not bound the performance.  For that, we rely on termination subroutines that ensure fast termination:

\begin{theorem}
If termination subroutine $T$ satisfies Fast Termination instead of Eventual Termination, and if it uses $s$ states and $q$ precision, and if it runs in time $t$, then the universal algorithm solves leader election with error $\epsilon$ with $s + O(1)$ states and $q$ precision (guaranteeing safety only in networks of size $\geq \nmin$).  Furthermore, it terminates in $O(t\log_{\hat{q}}(n + 1/\epsilon))$ rounds, with probability at least $1-\epsilon$.
\label{thm:universalfast}
\end{theorem}
\begin{proof}
Notice that Fast Termination implies Eventual Termination: eventually, there will be only one active node, since if there is more than one active node, then in every iteration of the knockout loop there is a positive probability that one node is knocked out; in the iteration immediately after the last knockout occurs, the conditions for Fast Termination are satisfied.
We can, therefore, apply
Theorem~\ref{thm:universal} to establish the algorithm as a correct leader election algorithm.  
It remains only to consider the running time.  There are two cases. 

Assume $n=1$.  In this case, in Lines 1--3, the lone node calls the termination subroutine with $ko$ and $active$ both equal true.  By the fast termination property, the lone node terminates before it even enters the knockout loop.

Assume $n > 1$.  By Lemma~\ref{lem:ko}, we know that with probability at least $1-\epsilon$, there is only one active node within $O(tR)$ time.  Consider the last round that begins with at least two active nodes.  In that round, at least one node is knocked out, and hence the termination routine is called with only one active node and at least one node with $ko$ equal to true.  By the fast termination property (and the agreement property), all nodes terminate after this call. 

To calculate the final time cost, we note that each iteration of the knockout loop requires $t + O(1)$ time: satisfying the claimed time complexity upper bound. 
\end{proof}

\subsection{Optimal Leader Election}
\label{sec:leader:stateoptimal}

Here we define a termination subroutine that, when combined with the universal leader election algorithm, matches our lower bound from Theorem~\ref{thm:lowerleader}.
In more detail, fix an error bound $\epsilon$ and probability precision $q$. Fix some lower bound $\nmin \geq 1$ on the network size.
We describe a termination detection subroutine that we call {\em StateOptimal}$(\nmin)$ that requires $O(\lceil \frac{\log_{q}{(1/\epsilon)}}{\nmin}\rceil)$ states, and guarantees Agreement, Termination, and Safety in any network of size $n \geq \nmin$.

There are two important points relevant to  this leader election strategy. First, for $\nmin = 1$, it provides a general solution that works in every size network. Second, the state requirements for this algorithm are asymptotically optimal according to Theorem~\ref{thm:lowerleader}.
As will be clear from its definition below, the cost of this optimality is inefficiency (its expected time increases exponentially with $n$). 
We will subsequently identify a pair of more efficient solutions that gain efficiency at the cost of some optimality under some conditions.

\paragraph{The {\em StateOptimal}$(\nmin)$ Termination Detection Subroutine.}
The {\em StateOptimal}$(\nmin)$ subroutine, unlike the other subroutines we will consider, ignores the $active$ and $ko$ parameters. Instead, it runs simple distributed coin flip logic among {\em all} nodes.
In more detail, recall from the definition of the universal algorithm
that $\hat q = \min\{q, (1/\epsilon)\}$.  
The subroutine consists of $\delta = \lceil \frac{c\log_{\hat q}{(1/\epsilon)}}{\nmin}\rceil$ rounds, defined
for some constant $c \geq 1$ we will bound in the analysis.
In each round, each node beeps with probability $1 - 1/{\hat q}$. At the end of the $\delta$ rounds, each node returns $1$ if all $\delta$ rounds were silent, otherwise it returns $0$.

\paragraph{Analysis.}
It is straightforward to determine that all nodes return the same value from this subroutine (i.e., if any node beeps or detects a beep, all nodes will return $0$).  It is also straightforward to verify that implementing this subroutine for a given $\delta$ requires $\Theta(\delta) = \Theta(\lceil \frac{\log_{\hat q}{(1/\epsilon)}}{\nmin}\rceil) = \Theta(\lceil \frac{\log_{q}{(1/\epsilon)}}{\nmin}\rceil)$ states (we can replace the $\hat q$ with $q$ in the final step
because once $q$ gets beyond size $1/\epsilon$, the function stabilizes at $1$).  
Eventual termination is also easy to verify, as every call to the subroutine has a probability strictly greater than $0$ of terminating.


To show safety, we observe that the routine returns true only if all $n$ nodes are silent for all $\delta$ rounds.  The probability of this happening is exponentially small in $(\delta n)$ and hence it is not hard to show that every $R$ invocations, the probability that the subroutine returns true in any invocation with more than one active node is at most $\epsilon/2$.
%
\begin{lemma}[Safety]
Over the first $R$ invocations, the probability that the subroutine returns true in any invocation with more than one active node is at most $\epsilon/2$.
\end{lemma}
\begin{proof}
The termination routine returns true only if all $n$ nodes are silent for all $\delta$ rounds.
Because each node beeps with probability $1/\hat{q}$,  the probability of this occurring is at most $1/\hat{q}^{\delta n}$.
There are now two cases to consider.  
The first case is when $\epsilon > 1/n$.  We observe here that $\delta \geq c$ (due to the ceiling) for some constant $c$, and so $1/\hat{q}^{\delta n} \leq 1/2^{nc}$.  Taking a union bound over the $R = 4\log_{\hat{q}}(n)$ invocations, we conclude that the probability of violating safety is $R/2^{nc} \leq 1/n^2 \leq \epsilon/2$ for a proper choice of $c$.
The second case is when $\epsilon \leq 1/n$. We observe here that $\delta = c\log_{\hat{q}}(1/\epsilon)/\nmin$, and so $1/\hat{q}^{\delta n} \leq \epsilon^c$.  We now take a union bound over $R = 4\log_{\hat{q}}(\max(1/\epsilon))$ calls to the subroutine, and conclude that the probability of violating safety is $R\epsilon^c \leq \epsilon^2 \leq \epsilon/2$, for proper choice of $c$.
\end{proof}

Combined with Theorem~\ref{thm:universal}, this yields the following conclusion:
\begin{theorem}
For any network size lower bound $\nmin$, error parameter $\epsilon$ and precision $q$, the universal leader election algorithm combined with the {StateOptimal}$(\nmin)$ subroutine, solves leader election with respect to these parameters when run in a network of size $n \geq \nmin$, and requires only $s = \Theta(\lceil \frac{\log_{q}{(1/\epsilon)}}{\nmin}\rceil)$ states.
 \label{thm:stateoptimal}
\end{theorem}

\subsection{Fast Leader Election with Sub-Optimal State}
\label{sec:fast}

The leader election algorithm from Section~\ref{sec:leader:stateoptimal} can solve the problem with the optimal amount of states for any given combination of system parameters. It achieves this feat, however, at the expense of time complexity: it is straightforward to determine that this algorithm requires  time exponential in the network size to terminate. 
%
Here we consider a termination subroutine that trades some state optimality  to achieve a solution that is fast (polylogarthmic in $1/\epsilon$ rounds) and simple to define (it uses the minimal probabilistic precision of $q=2$). Furthermore, its definition is independent of the network size $n$, yet it still works for every possible $n$. 
For the purpose of this section, we assume that $q = \hat{q} = 2$.
As we show below, this subroutine uses $\Theta(\log{(1/\epsilon)})$ states. This is suboptimal when high precision (i.e., larger $q$) is available, and when there is a lower bound $\nmin$ on the size of the network.

\paragraph{The {\em Fixed Error} Termination Detection Subroutine.}
This termination subroutine consists of a fixed schedule of $\lceil \log{(2/\epsilon)} \rceil + 2$ rounds. During the first round, any node that calls the subroutine with parameter $ko$ equal to $1$ beeps while all other nodes receive. If no node beeps, then the subroutine is aborted and all nodes return  false.

Moving forward, assume this event does not occur (i.e., at least one node beeps in the first round). For each of the $\lceil \log{(2/\epsilon)} \rceil$ rounds that follow, every node with parameter $active = 1$, will flip a fair two-sided coin. If it comes up heads, it will beep, otherwise it will receive. Each node with $active = 1$ will start these rounds with a flag $solo$ initialized to $1$. If such a node ever detects a beep during a round that it receives, it will reset $solo$ to $0$ (as it just gained evidence that it is not alone). 

The final round of the subroutine is used to determine if anyone detected a non-solo execution of the subroutine. 
To do so, every node with $active =1$ and $solo=0$ beeps. If no node beeps in this final round, then all nodes return true. Otherwise, all nodes return false.

\paragraph{Analysis.}
We proceed as before, observing that all nodes return the same value from this subroutine since all observe the same channel activity in the first and last rounds.  It is also straightforward to verify that implementing this subroutine requires $O(\log{(1/\epsilon)})$ states to count the rounds and record $solo$.  Fast termination follows directly from a case analysis of the algorithm.
%

\begin{lemma}[Fast Termination]
If the {\em Fixed Error} subroutine is called with only one active node, and with at least one node where $ko = true$, then it returns true.
\label{lem:fasttermination}
\end{lemma}
\begin{proof}
In this case, the node with $ko = true$ beeps in the first round, ensuring that everyone continues.  During the following $\log{(1/\epsilon})$ rounds, only the single active node beeps, and hence it terminates with $solo=1$.  Thus, no node beeps in the final round and everyone returns true.
\end{proof}

Safety requires a little more care, showing that the failure probabilities over $R$ invocations can be bounded by $\epsilon/2$, since the error probability depends on the number of active nodes.  
\begin{lemma}[Safety]
\label{lem:safetyfast}
Over the first $R$ invocations of the subroutine, the probability that it returns true in any invocation with more than one active node is at most $\epsilon/2$.
\end{lemma}
\begin{proof}
First, we note that if there are $k$ nodes active, then in order to return true, all $k$ nodes must maintain $solo = 1$. If in any round, any two nodes take different actions---i.e., one beeps and the other is silent---then one will detect it is not solo and all nodes will subsequently return false.  Thus, the probability that every active node maintains $solo = 1$ is at most $1/2^{sk}$, where $s = \lceil{\log{2/\epsilon}}\rceil$.	

We now consider the $R$ invocations of the subroutine. 
 In any invocation where no node has $ko = true$, all the nodes always return false.  We can therefore
  focus only on the subset of $t \leq R$ invocations where at least one node has $ko = true$.
Let $k_1, k_2, \ldots, k_t$ be the number of active nodes in each invocation.  In any invocation where $ko = true$, we know that the number of active nodes is reduced by $1$ as compared to the last invocation. 
 We can therefore conclude that $k_1 > k_2 > k_3 > \ldots > k_t$.

By a union bound, the probability that we return true in any of these invocations, given that there are more than $1$ active nodes, is at most:
\begin{eqnarray*}
(1/2)^{k_1s} + (1/2)^{k_2s} + ... + (1/2)^{k_ts} &  < &  (1/2)^{k_1}(1/2)^{s} + (1/2)^{k_2}(1/2)^{s} + ... + (1/2)^{k_t}(1/2)^s \\
   &   = & (1/2)^s \big[  (1/2)^{k_1} + (1/2)^{k_2} + ... + (1/2)^{k_t} \big] \\
   &  \leq  & (1/2)^s \big[ (1/2)^t + (1/2)^{t-1} + ... + (1/2) \big] \\
   & < & (1/2)^s \leq \epsilon/2
\end{eqnarray*}
\end{proof}
%
%
Combined with Theorem~\ref{thm:universalfast}, these properties yield the following conclusion:

\begin{theorem}
For error parameter $\epsilon$, the universal leader election algorithm combined with the 
{Fixed Error} subroutine, solves leader election with respect to $\epsilon$
in every size network, using only $s=\Theta(\log{(1/\epsilon)})$ states and $q=2$.
With probability at least $1-\epsilon$, it terminates in $O(\log{(n + 1/\epsilon)}\log{(1/\epsilon)})$ rounds. %
 \label{thm:stateoptimalfast}
\end{theorem}

\subsection{Fast Leader Election with $O(1)$ States and High Probability}
\label{sec:faster}

The final termination detection subroutine we consider
requires only a constant number of states,
and when executed in a network of size $n$,
for any $n > 1$,
it solves leader election with high probability in $n$.
At first glance, this result may seem to violate the lower bound from Section~\ref{sec:leader:lower},
which notes that the state requirement grows with a $\log{(1/\epsilon)}$ factor as $\epsilon$ decreases.
The question is why a constant number of states is sufficient for our algorithm here even though this term grows
with $n$.
The answer relies on the fact that $\epsilon$ is always a function of $n$,
such that for any fixed $n$, it is true that $\nmin \geq n$, and therefore the $\nmin$ factor
in the denominator of our lower bound swamps the growth of the $\log{n}$ factor in the numerator.

\paragraph{The {\em Constant State} Termination Detection Subroutine.}
The subroutine here is identical to the {\em Fixed Error} subroutine, except the length of subroutine is not fixed in advance
(no node has enough states to count beyond a constant number of rounds---which is not enough for our purposes).
  Instead, we dynamically adapt the length
  of the subroutine to a sufficiently large function of $n$ using a distributed counting strategy.
  
In more detail, during the first round, any node that called the subroutine with parameter $ko$ equal to $1$ beeps while all other nodes receive. If no node beeps, then subroutine is aborted and all nodes will return value false (as is true for {\em Fixed Error}).
Assuming the subroutine has not aborted, the nodes then proceed to the main body.  During the main body,  we partition rounds
into even and odd pairs.
During the odd numbered rounds, we proceed as in {\em Fixed Error}: every node with parameter $active = 1$, flips a fair coin; if it comes up heads, it will beep, otherwise it will receive; 
each node with $active = 1$ will start these rounds with a flag $solo$ initialized to $1$; if such a node ever detects a beep during a round that it receives, it will reset $solo$ to $0$ (as it just gained evidence that it is not alone).

During the even rounds, the nodes run a repeated knockout protocol for $O(1)$ iterations, for some fixed constant bounded in our below analysis.  
In more detail,
each node (regardless of whether or not it has $active$ equal to true)
begins the subroutine with a flag $attack=1$ and a counter $count = 0$.  In each even round, each node with $attack=1$ flips a fair coin and beeps if it comes up heads; otherwise it listens.  Any node that listens in an even round and hears a beep sets $attack=0$.  If there is an even round in which no node beeps, then all nodes increment $count$ and reset $attack = 1$.  
This continues until $count$ grows larger than the fixed constant mentioned above,
When this occurs, all nodes move the final round, which is defined the same as the final round in {\em Fixed Error.}
That is: every node with $active =1$ and $solo=0$ beeps. If no node beeps in this final round, then all nodes return true. Otherwise, all nodes return false.

\paragraph{Analysis.}
The Liveness and Fast Termination properties follow from the same arguments used in our analysis of {\em Fixed Error}.
The main difficulty in analyzing this subroutine is proving Safety.
To do so, we first bound how long the subroutine is likely to run on any given invocation:

\begin{lemma}
\label{lem:knockoutlength}
For any constant $c$, there exists a $c' > c$ and a constant bound for $count$,
such that the main body of the subroutine runs for at least $c\log(n)$ rounds but no more than $c'\log{n}$ rounds, with high probability.
\end{lemma}
\begin{proof}
We will show that with high probability, each iteration of the knockout protocol runs for $\Theta(\log{n})$ rounds.  By choosing 
an appropriate constant bound for $count$ (i.e., the number of iterations of the knockout protocol), 
we can increase this length by any constant factor.

We begin by noting that as long as there are $\Theta(\log{n})$ nodes with $attack=1$, then, with high probability, the iteration does not end
in that round, as it is likely that at least one node will beep (this follows because the probability
of silence in this case is upper bounded by $1/2^{\Theta(\log{n})}$).
Also note that whenever a node flips a tails, it is eliminated.  Thus, we show that with high probability, at least $\Theta(\log{n})$ nodes flip $\Theta(\log(n/\log\log{n}) = \Theta(\log{n})$ heads in a row.
 
One easy way to to see why this is so is as follows.  Divide the nodes in $\Theta(\log{n})$ different groups of $n/\Theta(\log{n})$ nodes each.  For each group, the probability that no node flips sufficiently many heads in a row is at most: 
\begin{eqnarray*}
\left(1 - \frac{1}{2^{\log(n/\log{n})}}\right)^{n/\Theta(\log{n})} & \leq & \left(1 - \frac{\log{n}}{n}\right)^{{n/\Theta(\log{n})}} \\ 
& \leq & (1/e)^{\Theta(1)}
\end{eqnarray*}
Thus each group, independently, has a constant probability of having one node survive.  Since each group is independent and there are $\Theta(\log{n})$ groups, by a Chernoff bound there are $\Theta(\log{n})$ survivors with high probability.
Thus we conclude that with high probability, the subroutine does not terminate for $\Theta(\log(n/\log{n})) = \Theta(\log{n})$ rounds.

To show that the subroutine does eventually terminate for some larger term in  $\Theta(\log{n})$ is more straightforward.
Fix a specific node $u$. 
For every even round such that $u$ has $attack=1$,
there is a probability of at least $1/2$ that it receives a beep (and resets $attack$) 
or there is a silence.
It follows that with high probability one of these two things has happened after $\Theta(\log{n})$ rounds.
A union bound over all nodes shows that with high probability this is true for all nodes, indicating
either a silence has occurred or every node has $attack=0$ so a silence is about to occur.
\end{proof}

\begin{lemma}[Safety]
\label{lem:safetyfaster}
Over the first $R$ invocations of the subroutine, the probability that it returns true in any invocation with more than one active node is at most $1/n^c$, for a  constant $c$ we can grow with our constant bound on $count$.
\end{lemma}
\begin{proof}
The proof here proceeds much as in Lemma~\ref{lem:safetyfast}.  As before, we can conclude that the probability of returning true when there is more than one active node is bounded by $1/2^s$, where $s$ is the number of round of the main body of the subroutine.  We have shown in Lemma~\ref{lem:knockoutlength} that $s = \Omega(\log{n})$ with high probability, and so overall we conclude that the lemma holds with high probability.
\end{proof}
%

We can then show that the subroutine guarantees safety. Combined with the Theorem~\ref{thm:universalfast}, these properties yields the following conclusion:

\begin{theorem}
For any network size $n$, the universal leader election algorithm combined with the 
{Constant State} termination detection subroutine, solves leader election
with high probability in $n$ using $s=O(1)$ states and $q=2$.
Also with high probability in $n$,
it terminates in $O(\log^2{n})$ rounds.
 \label{thm:stateoptimalfaster}
\end{theorem}

\section{Solving General Distributed Decision Problems}

In the previous section, 
we studied upper and lower bounds for solving leader election in the beeping model.
Here we establish these leader election bounds to be (in some sense) fundamental
for useful distributed computation in this setting. 
In more detail, we use a combination of our leader election algorithms
as a key primitive in constructing an algorithm that
can simulate a logspace (in $n$) decider Turing Machine (TM) with a constant number of unary
input tapes (of size $O(n)$ each).
The simulation has  error probability at most $\epsilon$,
 requires only the minimum probabilistic precision ($q=2$),
 and uses $s =O(\log{(1/\epsilon)})$ states.
If high probability in $n$ is sufficient, then
the state size can be reduced to $s=O(1)$.
Formally:

\begin{theorem}
For any problem solvable by a logspace TM with a constant number of unary input tapes,
there exist constants $c,d\geq 1$,
such that for any error probability $\epsilon \in [0,1/2]$ and network size $n\geq 1$,
we can solve the problem in the beeping model in a network of size $n$ with probability at least $1-\epsilon$
using $s = c\log{(1/\epsilon)}$ states, precision $q=2$, and an expected running time of $O(n^d\log^2{(n+ 1/\epsilon)})$ rounds.
For high probability correctness, $s=O(1)$ states are sufficient.
\label{thm:tm}
\end{theorem}

\noindent We now highlight some important observations about the above result.
 {First,}  we should not expect to simulate a {\em more} powerful class of TM.
This follows from a configuration counting argument.
For $s=O(1)$, for example, the $n$ nodes in our model can collectively encode at most $O(n^s)$ unique configurations
(there is no explicit ordering of nodes, so a given configuration of our system is described by the number of nodes out of $n$
in each of the $s$ possible states). 
A TM with more than log space,  by contrast, might have many more possible configurations that need to be simulated. 
What is perhaps more surprising is that we can successfully simulate a logspace machine even though nodes do not
have enough states for unique ids or even to store a single pointer to the simulated TM's tape.
 In some sense, our algorithm is making
full use (asymptotically) of the available memory in our distributed system.

Second, notice that the size of this algorithm is independent of the network size.
The same number of states successfully simulates the TM even
as $n$, and therefore the potential length of the simulated TM computation,
grows to arbitrarily large values.
Third, this results establishes $\log{(1/\epsilon)}$ as a key state complexity threshold in the beeping model. 
If you have fewer than this many states, 
you cannot even safely solve basic symmetry breaking tasks (e.g., leader election).
Once you reach this threshold, however, suddenly you can solve a rich set of expressive
problems (e.g., anything solvable by a logspace TM).

Finally, we emphasize that we do not present this simulation as a practical algorithm for solving problems in limited communication scenarios
(simulating a TM typically adds many more layers of indirection than is necessary).
We instead use this result to identify the threshold beyond which beeping nodes can start to solve interesting problems.
Finding elegant solutions to individual problems in this class is a separate and useful endeavor.

Before proceeding to the proof details, we first
summarize the main ideas.
Our result depends on a TM simulation strategy that follows the  outline
originally identified in~\cite{AngluinADFP06},
where it was used to simulate a TM using a population protocol in the randomized interaction model.
 In more detail, we first simulate a simple counter machine with a constant number
 of counters that can take values of size $O(n)$. 
 We then apply a classical computability result due to Minsky~\cite{minsky1967} which
  shows how to simulate a logspace TM (with unary input tapes) using a counter machine of this type.
The counter machine simulation in the beeping model,
combined with Minsky's TM simulation in the counter model,
yields a TM simulation in the beeping model.
Though we follow the same outline as in~\cite{AngluinADFP06},
the details of our counter simulation of course differ as we are implementing
this simulation in the beeping model whereas \cite{AngluinADFP06} implements
their simulation in the population protocol model.
What our two approaches do share (along with many network simulations of TMs)
is the use of leaders to coordinate the simulation.

The core concept in our counter machine simulation is to elect a leader
to play the role of the simulation coordinator.
This coordinator can announce the counter operations that need to be simulated by the network.
We show for every operation required of a counter machine,
there is a way to simulate its operation with at most  a constant number of leader election instances.
To elect leaders, we use our universal algorithm combined with the {\em conjunction} of both our fast termination detection subroutines
(i.e., termination requires both to return true). We show that the error probability of this combination is bounded by $\epsilon/n^c$---allowing us
to safely solve leader election for the needed polynomial number of instances before the simulated logspace TM reaches its final decision.
We now tackle these elements in more detail.

\paragraph{Simulating a Counter Machine.}
The counter machine we simulate has access to a constant number of counters
that can hold values from $0$ to $O(n)$.
Control is captured by a finite state machine.
Each state transition can increment, decrement, or reset to $0$ any of the counter values.
The transition function can also integrate the outcome of a comparison operator
that compares the value stored in a given counter to $0$. 

Our simulation leverages leader election as a key subroutine.
In particular, we use the primitive that results when we combine the universal leader election
algorithm with the the conjunctive combination of {\em both} the {\em Fixed Error} and {\em Constant State}
termination detection subroutines.
That is, every time the universal algorithm calls a termination subroutine, it will now call both of these subroutines, 
one after the other, and then return the value $out_1 \wedge out_2$,
where $out_1$ and $out_2$ are the outcomes of the first and second subroutine called, respectively.
For the remainder of this discussion, we call this instantiation of the universal algorithm the {\em double-safe fast leader election} algorithm.
Because the back-to-back executions of these subroutines are independent, 
and they offer error bounds of $\epsilon$ and $n^{-c}$ (for some constant $c>1$ that grows with the available state), respectively,
we get the following claim about this algorithm:

\begin{claim}
The probability that a given call to the termination subroutine fails in the double-safe leader election algorithm
is no more than $\epsilon/n^c$.
\end{claim}

Returning to our simulation description, the first step is to run the double-safe protocol to elect a leader to play
the role of {\em coordinator}.
This node is responsible for simulating the state transitions of the finite-state control of the counter machine.
It is also responsible for announcing to the other nodes (using predetermined, constant length beep patterns) which operation will be simulated next.

We make use of the states distributed among all $n$ nodes in the network to 
store the counter values in a distributed fashion.
To so so, we assume for each counter $c_i$, each node has a local bit labelled $c[i]$ in its state.
Our simulation will store counter values in unary using these bits.
That is, we represent $c_i = x$ at a given point in our simulation
by having exactly $x$ nodes in our beep model with their $c[i]$ bit set to $1$.
(To handle counter values larger than $n$, but still in $O(n)$, we can expand the size of these $c[i]$ local counters to larger constant sizes as needed).

Two of these counters (let us call them $c_1$ and $c_2$) are needed to run the TM simulation, and we assume are initialized to zero.
Accordingly, we assume all nodes begin with $c[1] = c[2] = 0$.
The simulation also assumes the values stored in unary on the input tape(s) of the simulated TM are also initially stored in counters.
We can capture this in a similar manner; e.g., if $c_3$ corresponds to an input tape storing value $x$ in unary, we
assume exactly $x$ nodes begin with $c[3]=1$.

We now describe each of the operations that may need to be simulated, and show
how the nodes can successfully simulate each of these operations using only $O(1)$ correct calls to double-safe leader election.
The coordinator, true to its name, coordinates these operation simulations.
That is, it announces the next operation to be simulated with a fixed beep pattern.
It then uses the results of operation simulation to advance its local copy of the counter machine control,
which determines which operation to simulate next.

\begin{itemize}

	\item {\em Increment.} To increment a counter $c_i$, the set of nodes (if any) with $c[i]=0$ run leader election.
	The winner sets its $c[i]$ bit to $1$. (If all nodes have $c[i]=1$ then no election is necessary as the counter does not grow beyond its maximum value.)
	
	\item {\em Decrement}. To decrement a counter $c_i$, 
	we follow the same strategy as the increment operation, except now nodes with $c[i]=1$ compete, and the leader resets $c[i]$ to $0$.
	(If all nodes have $c[i]=0$, then no election is necessary as the counter cannot reduce below $0$.)
	
	\item {\em Zero}. To zero a counter $c_i$,
	all nodes with $c[i]=1$ reset $c[i]$ to $0$.
	
	\item {\em Compare to Zero}. To compare a counter to $c_i$ to $0$,
	it is sufficient to assign a round for all nodes with $c[i]=1$ to beep.
	If the coordinator detects silent then it knows the stored counter value is $0$, otherwise it is greater than $0$.

\end{itemize}

\paragraph{The TM Simulation}
Though the counter machine described above is quite simple,
Minsky~\cite{minsky1967} shows it is sufficiently powerful to simulate
a logspace TM with a constant number of unary input tapes.
Minsky's simulation requires up to a polynomial number of steps
of the counter machine for each simulated step of the TM.\footnote{Minsky's simulation stores the working tape
of the simulated TM has a base-$b$ value (where $b$ is tape alphabet size) stored in unary in the counters. 
In particular, the value is split between two counters. Step simulations require multiplication and division operations,
which, to implement using the increment and decrement operations available in the counter machine,
can require steps linear in the counter sizes.}
Therefore, our TM simulation will require up to a polynomial number of successful leader election
calls per simulated TM step.
The total number of simulated TM steps can also be bounded by a polynomial,
as the machine has only logarithmic space and it is a deterministic decider.
Therefore, our simulation must correctly implement leader election a (larger) polynomial
number of times to correctly simulate the TM until its decision.

\paragraph{Analysis.}
We note that a sufficiently large constant number of states is enough for the coordinator to simulate the finite control 
of the counter machine.
We also note that an additional constant number of states
is enough for the nodes to store their constant-sized pieces of the constant number of counter values used in the simulation.

More interesting is the question of how many states are needed to ensure that the leader elections calls in the simulation are all correct.
To answer this question, let $n^a$, for some constant $a \geq 1$ dependent on the TM definition and Minsky simulation details,
be the maximum number of leader election calls our simulation might make.
We note that our double-safe algorithm uses two termination detection subroutines.
The first requires $s=O(\log{(1/\epsilon)})$ states to reduce the error probability to no more than $\epsilon$.
The second requires $s=O(1)$ states to reduce the error probability to $n^{-c}$, for some $c$ we can grow by increasing the constant in the state size.
If we fix $c \geq a$, we see the probability that a particular call to leader election fails
is no more than $\epsilon/n^a$.
A union bound over the no more than $n^a$ leader elections needed by the simulation tells
us that the probability at least one fails is less than $\epsilon$.

Finally, we turn our attention to time complexity.
Because each of the fast termination detection subroutines terminate in time $O(\log^2{(n+1/\epsilon)})$,
we get $O(n^a\log^2{(n+1/\epsilon)})$ as an expected time bound.
These results combine to establish Theorem~\ref{thm:tm}.




\bibliographystyle{plain}
\bibliography{bib}

\appendix

\end{document}